\documentclass{sig-alternate-05-2015}
\usepackage{amsfonts,amssymb,amsmath,latexsym,ae,aecompl}
\usepackage{float}
\usepackage{epsfig}
\usepackage{enumerate}
\usepackage{graphicx}
\usepackage{fourier}
\usepackage{algorithmic}
\usepackage{caption}
\usepackage{graphicx}
\usepackage{subfig}
\usepackage{gensymb}
\usepackage{color}
\newtheorem{theorem}{Theorem}
\newtheorem{lemma}{Lemma}
\newtheorem{pdefinition}{Problem}
\newtheorem{behavior}{Behavior}
\newtheorem{observation}{Observation}
\usepackage[ruled,vlined,linesnumbered]{algorithm2e}

\begin{document}

\setcopyright{acmcopyright}

\conferenceinfo{ICDCN '17}{January 4--7, 2017, Hyderabad, India}


%
\conferenceinfo{ICDCN '17}{January 4--7, 2017, Hyderabad, India}


\title{Fault-Tolerant Gathering of Mobile Robots with \\Weak Multiplicity Detection
}

\numberofauthors{4}
\author{
\alignauthor
Debasish Pattanayak\\
       \affaddr{IIT Guwahati, India}\\
       \email{p.debasish@iitg.ernet.in}
\alignauthor
Kaushik Mondal\\
 \affaddr{IIT Guwahati, India}\\       			
       \email{mondal.k@iitg.ernet.in}
\alignauthor
H. Ramesh\\
        \affaddr{IIT Guwahati, India}\\
       \email{ramesh\_h@iitg.ernet.in}
\and  
\alignauthor Partha Sarathi Mandal\\
      \affaddr{IIT Guwahati, India}\\
       \email{psm@iitg.ernet.in}
}
\date{05 August 2016}
\maketitle

\begin{abstract}
There has been a wide interest in designing distributed algorithms for tiny robots. In particular, it has been shown that the robots can complete certain tasks even in the presence of faulty robots.
In this paper, we focus on gathering of all non-faulty robots at a single point in presence of faulty robots.
We propose a wait-free algorithm (i.e., no robot waits for other robot and algorithm instructs each robot to move in every step, unless it is already at the gathering location), that gathers all non-faulty robots in semi-synchronous model without any agreement in the coordinate system and with weak multiplicity detection (i.e., a robot can only detect that either there is one or more robot at a location) in the presence of at most $n-1$ faulty robots for $n\geqslant 3$.
We show that the required capability for gathering robots is minimal in the above model, since relaxing it further makes gathering impossible to solve.

Also, we introduce an intermediate scheduling model \textit{ASYNC$_{IC}$} between the asynchronous ( i.e., no instantaneous movement or computation) and the semi-synchronous (i.e., both instantaneous movement and computation) as the asynchronous model with instantaneous computation.
Then we propose another algorithm in \textit{ASYNC$_{IC}$} model for gathering all non-faulty
robots  with weak multiplicity detection without any agreement on the coordinate system in the presence of at most $\lfloor n/2\rfloor-2$ faulty robots for $n\geqslant 7$.
\end{abstract}

%
%

\section{Introduction}
Distributed coordination among robots in multi-robot systems has garnered interest in recent years.
These multi-robot systems have applicability in various fields such as space exploration, disaster rescue, exploration, military operations, etc. The primary motivation in this field is to find the minimum capability required to achieve certain objectives for a system of robots. Over the course of study, various robot models has been used.
Specifically, "weak robots" \cite{FlocchiniPSW99} is the most widely considered model. Weak robots are \textit{autonomous}: behave independently, \textit{anonymous}: do not have identifiers, \textit{oblivious}: do not remember their past actions and \textit{silent}: do not exchange messages among each other. Mostly, they do not follow a common coordinate system. The robots are represented as points in a plane. They may have the capability of \textit{multiplicity detection}, i.e. a point having multiple robots.
A robot with \textit{weak} multiplicity detection can only figure out whether a point is occupied by exactly one robot or more than one robot.
Similarly with \textit{strong} multiplicity detection a robot can detect the exact number of robots at a point.
 They have either limited or unlimited visibility range. The robots are either transparent or non-transparent.
 Each robot follows \textit{look-compute-move} cycles.
 It observes the surrounding
in the \textit{look} phase. In the \textit{compute} phase it computes
the destination based on the observation. It moves towards the
destination point in the \textit{move} phase. In the \textit{semi-synchronous} (\textit{SSYNC}) model, the global time is divided into discrete time intervals
called rounds. In each round a subset of robots are activated.
Once a robot is activated, it finishes one \textit{look-compute-move}
cycle in that round. The \textit{fully-synchronous} (\textit{FSYNC}) model can be considered a special case of \textit{SSYNC}, since it activates all the robots in each round.
In the \textit{asynchronous} (\textit{ASYNC}) model, any robot can be activated at any time. A robot can be idle for unpredictable but finite amount of time.
We consider the scheduler to be a fair scheduler, which activates the robots infinitely many times in infinite time.

Like any distributed system, there is a possibility of failure for any robot in these muti-robot systems also.
The faults considered are of mainly two types: {\it crash faults} and {\it Byzantine faults}:
In crash faults, the faulty robots stop moving. In Byzantine faults, the faulty robots behave in an unpredictable manner \cite{agmon2006fault}.
 The faults can occur due to the unreliable components in the system, which are custom made or manually built. Also sometimes there can be defect in manufacturing.
 Apart from these, the faults can be caused by some external factor in the field of deployment.
Hence there is a need to design fault-tolerant algorithms. In this paper we design the algorithms to be crash-fault tolerant.

\begin{table*}
\centering
\begin{tabular}{|c|c|c|c|c|c|c|}\hline
&&&&&&\\
\textbf{Reference} &\textbf{Model} &\textbf{ \# Faults} & \textbf{Direction} & \textbf{Chirality} & \textbf{Multiplicity Detection} & \textbf{Valid Initial Configuration} \\
&&&&&&\\\hline
\cite{agmon2006fault}&\textit{SSYNC} & 1 & NO & NO & Weak & at most one multiplicity  \\\hline
\cite{Bouzid0T13}& \textit{SSYNC} & $n-1$ & NO & YES & Strong & not bivalent \\\hline
 \cite{BramasT15}&\textit{SSYNC} & $n-1$ & NO & NO & Strong & not bivalent \\\hline
 Our Result & \textit{SSYNC} & $n-1$ & NO & NO & Weak & at most one multiplicity \\\hline
  \cite{Bhagat201650}&\textit{ASYNC} & $n-1$ & YES & NO & NO & any \\\hline
 Our Result & \textit{ASYNC}$_{IC}$ & $\lfloor n/2\rfloor -2$ & NO & NO & Weak & any symmetric or asymmetric configuration  \\
 & & & & & &  with at most one multiplicity excluding \\
 & & & & & &  symmetric $C(0)$, $C(1/k)$, $C(1/2)$ and $C(1/2+1/k)$\\\hline
\end{tabular}
\caption{Summary of Assumptions for Fault-Tolerant Gathering}
\label{table:fault}
\end{table*}
\subsection{Related Works}
Some of the common problems for these multi-robot systems include,
\textit{leader election}: all robots agree on a leader  among themselves  \cite{ChaudhuriM15,kim1995leader,stoller2000leader}, \textit{gathering}:
 all robots gather at a single point \cite{suzuki1999distributed}, \textit{convergence}: the robots come very close to each other \cite{cohen2006convergence} and \textit{pattern formation}: the robots imitate a given pattern on the
plane \cite{suzuki1999distributed}.
The gathering problem has been studied for different models, including fully synchronous (\textit{FSYNC}), semi-synchronous (\textit{SSYNC}) and asynchronous (\textit{ASYNC}).
 In \textit{FSYNC} model, the gathering problem has been solved without making any additional assumptions to the basic model \cite{Cohen2005}. In \cite{suzuki1999distributed}, impossibility of gathering for $n=2$ without assumptions on local coordinate system agreement for \textit{SSYNC} and \textit{ASYNC} is proved.
Also, for $n>2$ it is impossible to solve gathering without assumptions on either coordinate system agreement or multiplicity detection \cite{Prencipe2007}.
In \cite{cieliebak2012distributed}, Cieliebak et al. have studied gathering with multiplicity detection.  A practical implementation of non-transparent fat robots with omnidirectional cameras opened up several new algorithmic issues \cite{HonoratPT14}.
 Chaudhuri et al. \cite{ChaudhuriM15} have proposed a deterministic algorithm for leader election and gathering for transparent fat robots without common sense of direction or chirality. Common chirality is basically the common clockwise order.

In recent years, devising algorithms that achieve the goal even in the presence of a few faulty robots has piqued the interest \cite{neiger1990automatically}.  In \textit{SSYNC} model, Agmon et al. \cite{agmon2006fault} have proposed an algorithm to gather robots with at most one faulty robot. The various problems in crash and Byzantine fault model are explored and some of the existing results are compiled by Clement et al. \cite{DefagoP0MPP16}.
Auger et al. \cite{AugerBCTU13} have proved that it is impossible to converge  oblivious mobile robots if more than one half and more than one third of the robots exhibit Byzantine failures respectively. Bouzid et al. \cite{Bouzid0T13} have proposed a wait-free crash-fault tolerant gathering algorithm with robots having strong multiplicity detection and chirality.
Bramas and Tixeuil \cite{BramasT15} has proposed a wait-free gathering algorithm for robots with arbitrary number of faults, which removed the assumption of chirality, but still has strong multiplicity detection as opposed to weak multiplicity detection in \cite{agmon2006fault}. They conjecture that weak multiplicity detection can only solve gathering for distinct initial positions.
 Bhagat et al. \cite{Bhagat201650} have solved the problem of gathering in \textit{ASYNC} setting $(n,n-1)$ crash fault model in 2D under agreement on the direction and orientation of one axis. To the best of our knowledge there is no fault-tolerant algorithm in the \textit{ASYNC} model without any agreement in coordinate system. The various assumptions in the results we found along with our results for fault-tolerant gathering are summarized in Table~\ref{table:fault}.

\subsection{Our Contributions}
 In this paper we propose two gathering algorithms, where the robots do not share a common direction (unlike \cite{Bhagat201650}) or chirality (unlike \cite{Bouzid0T13}). The robots only have capability of weak multiplicity detection (as opposed to strong multiplicity detection in \cite{Bouzid0T13,BramasT15}). The relaxation of assumption is possible because of our strategy of moving the robots in non-intersecting circular paths, which avoids creation of multiple multiplicity points.
\begin{enumerate}
\item We propose a wait-free gathering algorithm, which achieves gathering in finite time without common sense of direction or chirality only with weak multiplicity detection. This algorihtm extends the result by Agmon and Peleg \cite{agmon2006fault} to tolerate $n-1$ faults in the same model.
\item  We prove the conjecture in \cite{BramasT15} that strong multiplicity is required for gathering in presence of more than one multiplicity points in the admissible initial configurations in \textit{SSYNC} model.
\item We initiate the study on an asynchronous scheduling model with instantaneous computation (\textit{ASYNC}$_{IC}$) which is an intermediate model between semi-synchronous and asynchronous.
\item We propose a fault-tolerant algorithm in the \textit{ASYNC}$_{IC}$ which can gather even if almost half the number of total robots are faulty.
\end{enumerate}
\subsection{Paper Organization}
The remainder of the paper is organized as follows. Section~\ref{sec:prel} introduces the necessary background.
Section~\ref{sec:impossibity} shows an impossibility result that gathering is unsolvable in presence of more than one multiplicity point in \textit{SSYNC} model.
Section~\ref{sec:ssyncgather} presents algorithms for gathering in \textit{SSYNC} model. Section~\ref{sec:asyncic} introduces \textit{ASYNC}$_{IC}$ model. Then section~\ref{sec:asyncgather} presents a gathering algorithm in the \textit{ASYNC}$_{IC}$ model, before we conclude in section~\ref{sec:conclusion}.
\section{Preliminaries}\label{sec:prel}
Some important notations are as follows.
\begin{itemize}
\item $R = \{r_1,r_2,\cdots,r_n\}$ be the set of $n$ point robots in a Euclidean plane. Robot $r_i$ occupies the position $p_i$ in the plane. A multiplicity point can have multiple robots.
\item The configuration $C = \{p_1,p_2,\cdots,p_n\}$, is a multiset, which represents the robots in the plane.
\item A configuration is a \textit{legal configuration} if it has at most one multiplicity point.
\item The smallest enclosing circle ($SEC$) of a configuration $C$ is denoted by $SEC(C)$.
\item $C^b$ denotes the set of robots on the boundary of $SEC$. \\ $C^b_{fix}$ is the set of robots on boundary which are fixed.
\item In the ($n,f$) crash-fault model, out of the $n$ robots at most $f$ robots are faulty.
\end{itemize}

	 The robots have following behaviors in the crash fault model.
\begin{behavior}\label{behav:mindist}
For a non-faulty robot $r_i$ at position $p_i$ moving towards the destination  $p^*$ in its activated cycle, if the distance between $p_i$ and $p^{*}$ is less than $S$, where $S$ is a constant, then the robot $r_i$ reaches $p^*$ in the same cycle. Otherwise the robot stops
at a point on the line $\overline{p_ip^*}$ which is at least at a distance $S$ away from $p_i$.
\end{behavior}
\begin{behavior}
 A robot may become faulty at any point of time.
\end{behavior}
Gathering of robots in a crash-fault model is the gathering of all non-faulty robots at one point in finite time.

\section{Impossibility}\label{sec:impossibity}

It is impossible to gather two robots in SSYNC without agreement in coordinate system \cite{suzuki1999distributed}. The adversary can always schedule the robots such that at the end there would be exactly two multiplicity points.
As the robots only have capability of detecting either it is a multiplicity or not, not the capability to find out exactly how many robots are there in a multiplicity, so both multiplicity points behave the same as two robots in SSYNC model. Hence the theorem follows.
\begin{theorem}\label{theorem:impossibility}
For a non-legal configuration, it is impossible to design a wait-free deterministic algorithm which gathers all the robots with weak multiplicity detection in SSYNC model.
\end{theorem}
\begin{proof}
Suppose a deterministic wait-free algorithm $\psi$, gathers all the robots with weak multiplicity detection without agreement in coordinate system  in \textit{SSYNC} model from a configuration with more than one multiplicity point in finite time.
Say $C_t$ be the final configuration. Then there is exactly one multiplicity point in $C_t$.
If $C_{t'}$, for some $t'<t$, has more than two multiplicity points and $C_t$ has only one multiplicity then the three or more multiplicity points are merged into one multiplicity point before time $t$.
Now the three of them had merged at one of the three points or at some other point to create a single multiplicity. The adversary can always choose to not activate the robots in one multiplicity point which is not the destination in that cycle. Then $C_t$ would have two multiplicities. Hence it is not possible to have a configuration which has only one multiplicity point starting from a configuration with more than one multiplicity points.

Now consider the configuration $C_{t'}$ with two multiplicity points. Say all other robots except multiplicity robots gather at either one of the two multiplicity or any other point. If they go to one of the multiplicity then, the number of multiplicity doesn't decrease and remains two. If they go to some other point, then if the other robots create another multiplicity point then, it again becomes more than two multiplicity configuration.

So we can always argue that, there would always be a configuration with two multiplicity  points during the execution of algorithm. Now, the configuration with only two multiplicity points is same as having only two robots in SSYNC model. So those two can not gather \cite{suzuki1999distributed}, which is a contradiction. Hence no such algorithm $\psi$ exists.
\end{proof}
\begin{observation}
If any algorithm is $n-1$ fault-tolerant, then it must be wait-free.
\end{observation}
If the algorithm is not wait-free, then the non-faulty robot waits for some other robots to move. If all of them are faulty, then it results in an indefinite wait cycle. So the algorithm has to be wait-free. From theorem~\ref{theorem:impossibility}, it is clear that any wait-free gathering algorithm for \textit{SSYNC} model can gather non-faulty robots only starting from a configuration with at most one multiplicity point.

\section{Gathering ($n,n-1$) Crash Fault in SSYNC Model}\label{sec:ssyncgather}
\subsection{Model}\label{subsec:ssyncmodel}
The robots are modeled as points in a Euclidean plane. Each
robot can observe the environment and determine its position
along with the positions of other robots in its local coordinate
system. The robots are autonomous, anonymous, oblivious, homogeneous and silent.
The robots have unlimited visibility. They have the capacity of
weak multiplicity detection. Each robot follows
an atomic \textit{look-compute-move} cycle.
 Here we have considered the crash-fault
model. If a robot becomes faulty then it stops functioning.
The robots follow \textit{SSYNC} scheduling.

The problem statement is as follows:

\begin{pdefinition}(($n,n-1$) Crash Fault):
 Given $n$ anonymous, homogeneous, oblivious, point robots with unlimited visibility in a legal initial configuration with no agreement in coordinate system, having the ability to detect multiplicity points in \textit{SSYNC} model. The objective is to achieve gathering of non-faulty robots in the $(n,n-1)$ crash fault system for $n \geq 3$.
\end{pdefinition}
\subsection{Algorithm and Correctness}\label{subsec:ssyncgather}
In this section we present the Algorithm~\ref{algo:gatherk} (\textsc{Gather\_k($C,r_i$)}), which gathers the non-faulty robots. In the \textit{SSYNC} model, all the robots activated in each cycle look at the same time. Hence the view of each robot is the same. If there is a multiplicity point in the configuration then all the robots go towards the multiplicity point. If there is no multiplicity point then they go towards the center of $SEC$. As any robot which is activated can move, the $SEC$ can change in subsequent rounds. But for each round, the $SEC$ is same for all robots and hence the destination point (i.e., the center of $SEC$) also remains invariant for that round. All activated robots move towards the destination point  in a straight or circular path such that their paths do not intersect except for the destination point as described in Algorithm~\ref{algo:movetodest} (\textsc{MoveToDest}), which ensures the creation of a single multiplicity point. Algorithm~\ref{algo:movetodest} uses Algorithm~\ref{algo:findtangent} (\textsc{FindTangent}) as a subroutine which returns a line, which should be tangent to the circular path.
 Once the multiplicity point is formed, all robots move towards the multiplicity point, which remains invariant until gathering is finished in a similar manner.

We prove the correctness of the algorithms in a number of lemmas followed by a theorem. First, Lemma~\ref{lem:raddec} shows that the radius of the $SEC$ decreases every time the center of $SEC$ changes.
Lemma~\ref{lem:del2} shows the minimum decrement of distance from center of $SEC$ for a robot which moved radially towards the center of $SEC$ in the previous round. Lemma~\ref{lem:mindist} shows the minimum radial distance covered by a robot, if it moves in a circular path for a distance $S$.
Lemma~\ref{lem:destgather} proves that all the paths of robots are non-intersecting and Lemma~\ref{lem:multgather}
ensures that all the robots can reach their destination in finite time. Then
Theorem~\ref{theorem:kfault} using the previous lemmas guarantees that Algorithm~\ref{algo:gatherk} gathers all the robots in finite time in \textit{SSYNC} model.

\begin{lemma}\label{lem:raddec}
If the center of $SEC$ is different in the two consecutive cycles, then the new radius is less than the previous one.
\end{lemma}
\begin{proof}
Without loss of generality consider $C$ is the configuration before movement of robots and $C'$ is the configuration after movement.
In Figure \ref{fig:raddec}, $SEC$ of $C$ denoted with dotted circle and $SEC$ of $C'$ denoted with dashed circle. Assume that radius of $SEC$ of $C'$ is same as radius of $SEC$ of $C$. Consider the position of robots in consecutive cycles. The robots have moved on or inside the $SEC$ from configuration $C$. We can argue that the robot positions in $C'$ are within the intersection of $SEC$s of $C$ and $C'$.
\begin{figure}[H]\centering
\includegraphics[height=0.4\linewidth]{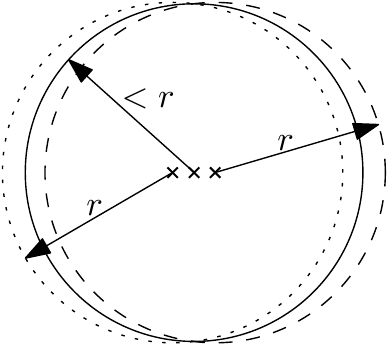}
\caption{Two different circles with same radius and a circle enclosing intersection of the two circles}\label{fig:raddec}
\end{figure}
But the smallest circle which can cover the intersection area of $SEC$s of $C$ and $C'$ is denoted with solid circle in Figure~\ref{fig:raddec} which has a radius less than $r$. Hence the $SEC$ of $C'$ denoted with dashed circle is not the $SEC$ of $C'$, which is a contradiction.
\end{proof}

\begin{lemma}\label{lem:del2}
If a robot moves distance $\delta$ towards the center of $SEC$ radially, then the distance from the center of $SEC$ in the next cycle is at most $\cfrac{r-\delta}{2} +\cfrac{1}{2}\sqrt{(r+\delta)^2-2\delta^2}$, where $r$ is the radius of $SEC$ in the current cycle.
\end{lemma}
\begin{proof}
Let $C$ and $C'$ be two configurations in consecutive cycles.
 As shown in Figure~\ref{fig:del-x}, the robot $r_i$ at $p$, represented as a black disk, after moving distance $\delta$ towards the center of $SEC$ in $C$ ends up on the boundary of $SEC$ in $C'$, denoted by $p'$, i.e. $pp' = \delta$. Say the centers of $SEC$ in $C$ and $C'$ are $O$ and $O'$ and $OO'=x$.
 If $p'$ is not collinear with $O$ and $O'$, then from triangle inequality $p'O' < p'O + OO' = r - \delta + x$.
We are trying to maximize $p'O'$ as it implies the distance of destination from the new position of robot.
The value $p'O'$ is maximum when the new position of robot is collinear with the centers i.e., the bisector of the chord made by the intersection points of two circle as shown in Figure~\ref{fig:del-x}. Then $p'O' = r - \delta +x$.\\
\begin{figure}[H]\centering
\includegraphics[height=0.4\linewidth]{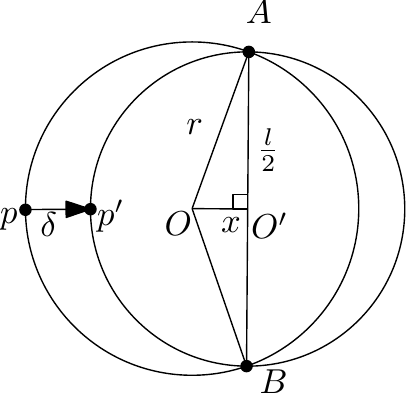}
\caption{Robot moving from boundary to boundary in consecutive configurations}\label{fig:del-x}
\end{figure}
Consider there are two robots on the intersection of both the circles at $A$ and $B$ in $C$ and $C'$. The largest $SEC$ possible with these three robots on the boundary is the circle with the other two robots on the ends of diameter.
 The largest $SEC$ radius can be achieved is $\frac{l}{2}$ where $l = AB$. From the right angled triangle $OO'A$

$$ r^2 = \frac{l^2}{4} + x^2 \implies \frac{l}{2} = \sqrt{r^2 - x^2}$$
Now
$$ r^2 - x^2 = (r-\delta+x)^2 \implies \delta - x = \frac{r+\delta}{2} -\frac{1}{2}\sqrt{(r+\delta)^2-2\delta^2}>0$$
Here $pO - p'O'=\delta-x$ is the minimum decrement of distance from center of $SEC$ for a robot which moved in previous round.
\end{proof}

\begin{algorithm}[!h]
\caption{\textsc{MoveToDest($C$,$p$,$p^*$,$style$,$distance$)}}\label{algo:movetodest}
\SetKwInOut{Input}{Input}\SetKwInOut{Output}{Output}
\Input{Robot position, destination point, movement style and distance}
\Output{Movement of robot}
\eIf{$style = straight$}
	{Move from $p$ to $p^*$ in a straight line}
	{$l$ = \textsc{FindTangent($C,p,p^*$)}\\
	Let $G$ be the circle passing through $p$ and $p^*$ with line $l$ as a tangent to $G$ at $p^*$.\\
	Move from $p$ to $p^*$ along the arc of circle $G$ in that sector.
	 }
\eIf{$distance = full$}{Move completely to $p^*$ along the path.}
	{Move along the path until the midpoint of path.}
\end{algorithm}
\begin{algorithm}[!h]
\caption{\textsc{FindTangent($C$,$p$,$p^*$)}}\label{algo:findtangent}
\SetKwInOut{Input}{Input}\SetKwInOut{Output}{Output}
\Input{Robot position, destination and configuration}
\Output{A line}
\eIf{there is a robot at $p'$ next to $p$ on $\overrightarrow{p^*p}$ away from $p^*$}
	{	
		$l$ = \textsc{FindTangent($C,p',p^*$)}\\
		Let circle $G$ passes through $p'$ and $p^*$ with $l$ as a tangent at $p^*$.\\
		\eIf{arclength of the minor arc is less than $S$ }
			{Return $l$}
			{Let $p''$ be the point on $G$ such that length of $\widearc{p'p''}$ is $S$\\
			Return $\overline{p''p^*}$}
	}
	{Find the robot free sector of $SEC$ adjacent to $p$ with smallest angle, say $\mu$.\\
	Return the angle bisector of $\mu$.}
\end{algorithm}

 \begin{algorithm}[!h]
\caption{\textsc{Gather\_k($C, r_i$)}}\label{algo:gatherk}
\SetKwInOut{Input}{Input}\SetKwInOut{Output}{Output}
\Input{A configuration $C$ and robot $r_i$}
\Output{Destination and movement path of robot.}
\eIf{there is a multiplicity point $m$}{$p^* = m$}
{$p^* = O$, where $O$ is the center of $SEC$}
\eIf{$r_i$ has a robot free path towards $p^*$}
{\textsc{MoveToDest($p_i,p^*,straight,full$))}}
{\textsc{MoveToDest($p_i,p^*,circular,full$))}}
\end{algorithm}
\begin{lemma}\label{lem:mindist}
For any non-faulty robot, the minimum radial distance traveled towards the center of $SEC$ in one cycle is at least $\lambda$, where $\lambda > \frac{S^2}{2r}$.
\end{lemma}
\begin{proof}
The minimum radial displacement towards the center of $SEC$ can be achieved when the arc of circle on whose boundary it travels has the highest curvature. The length of an minor arc between two points on a circle is highest when the arc is a semicircle.
\begin{figure}[H]\centering
\includegraphics[height=0.4\linewidth]{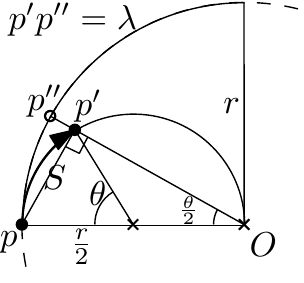}
\caption{Robot moving along a circular path towards destination point}
\label{fig:mindist}
\end{figure}
As shown in Figure~\ref{fig:mindist} the non-faulty robot travels at least distance $S$ on the circle from $p$ to $p'$. $p''$ is the projection of $p'$ onto boundary of $SEC$. Let the arc of distance $S$ on the inner circle makes an angle $\theta$ at the center of inner circle.
Then the arc  makes an angle $\frac{\theta}{2}$ at the center of $SEC$. We know that, $\frac{r}{2}\theta = S$. The semicircular angle is $\frac{\pi}{2}$. So $p'O = r\cos \frac{\theta}{2}$. From Figure~\ref{fig:mindist}, the radial displacement toward the center, i.e. $p'p''$ denoted by $\lambda$. Now $p'p'' = Op'' - Op'$, so

 $$\lambda = r - r\cos \frac{\theta}{2} = r\left(1 - \cos \frac{S}{r}\right) > r\left(1-\left(1-\cfrac{S^2}{2r^2}\right)\right)= \cfrac{S^2}{2r} $$
\end{proof}
\begin{lemma}\label{lem:destgather}
For any configuration Algorithm \ref{algo:gatherk} gathers all the non-faulty robots without creating any additional multiplicity point.
\end{lemma}
\begin{proof}
From Algorithm \ref{algo:gatherk}, notice that all the robots move towards the destination point to gather. If a robot does not have any other robot in its path then it moves in a straight line towards the destination. So the path it follows is unique.

\begin{figure}[H]\centering
\includegraphics[height=0.4\linewidth]{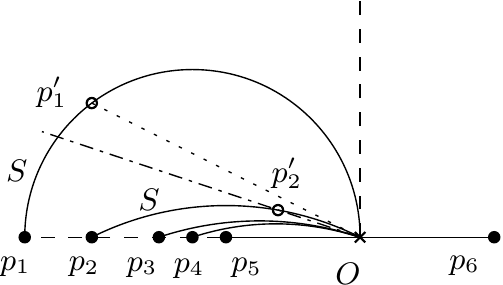}
\caption{Paths of robots intersect only at the destination point}\label{fig:destgather}
\end{figure}
A robot, whose straight line path is blocked by another robot, moves along a circular path connecting its current position and the destination point.
Observe from Figure \ref{fig:destgather} that the circular paths of robot at $p_1$ intersects with each other or the straight path of $p_5$ only at the destination point, because the circle is drawn such that  the angle bisector of $\angle p_5Op_6$ is the tangent to the circle at $O$. The robot at $p_2$ considers the line $p_1'O$ as the tangent where, $p_1'$ is the position of robot at $p_1$ after moving $S$ on the circular path. As the robot at $p_1$ will move straight in the next cycle, that straight line has been considered as the tangent for the circular path of robot at $p_2$.
For robots at $p_3$ and $p_4$ is the tangent is $p_2'O$, because the distance along circular path for $p_3$ is less than $S$, hence that will reach the destination in one cycle if activated.

 Hence it is not possible for the robots to form a multiplicity point at a point other than the destination point.
\end{proof}

\begin{lemma}\label{lem:multgather}
For a configuration with a multiplicity point, Algorithm~\ref{algo:gatherk} gathers at the multiplicity point in finite time.
\end{lemma}
\begin{proof}
According to Algorithm~\ref{algo:gatherk}, the destination point is the multiplicity point. The destination point remains invariant over the execution of algorithm, because from Lemma~\ref{fig:destgather}, the robots move in such a way that they only meet at the destination point without creating any other additional multiplicity point. From the Behaviour~\ref{behav:mindist} of non-faulty robots, they move at least a distance $S$ in each cycle they are activated. Hence the robots having a free corridor towards the destination will reach the destination in finite time.

Also notice that, for $k$ robots on a same line in one cycle,  which-ever robot moves they will not be collinear again. Say the robots are at a distance $\{d_1,d_2, \cdots,$ $ d_k\}$ from the center of $SEC$ in the increasing order of distance. After the robot $r_i$ moves $S$ distance along the circular arc, the angle it makes at the center of $SEC$ is half of the angle it makes at the center of circle along which it is moving. Hence the angle can be represented as, $\alpha_i = S \sin \theta/d_i$, where $\theta$ is the half of smallest angle made by a sector in $SEC$ adjacent to $r_i$ without any robots. Now for any two robots on the line, after the movement along the circular path, they will not be co-linear with the center of $SEC$, because the $\alpha_i$ will be different for each of them. Now each robot will have a robot free path towards the center of $SEC$. So after a finite number of activated rounds for that robot it will reach the destination point.
\end{proof}
\begin{theorem}\label{theorem:kfault}
For a legal initial configuration, Algorithm~\ref{algo:gatherk} gathers the non-faulty robots in finite time for $n\geqslant3$.
\end{theorem}
\begin{proof}
As Algorithm~\ref{algo:gatherk} works in \textit{SSYNC} model, in each cycle the configuration for all the robots are same. For a particular configuration if there is a multiplicity point then all the robot gather at the multiplicity point without creating additional multiplicity point on the way from Lemma~\ref{lem:multgather}.\\
If there is no multiplicity point initially in the configuration $C$, then each robot has to move towards the center of $SEC$. If the center of $SEC$ remains invariant across multiple rounds, then the robots activated during those rounds would form a multiplicity point at the center. For any initial configuration, the center of $SEC$ would always lie inside the convex hull of the configuration. Hence the total distance needed to traverse for a particular robot is finite.
 Hence in finite number of activation cycles there would be at least two robots which will form a multiplicity point. Also from Lemma~\ref{lem:del2}, whenever the $SEC$ is changing by the movement of a robot, it gets closer to the destination point in the next cycle. As the decrement is a positive finite value, inversely dependent on $r$, radius of $SEC$, the $SEC$ would gradually get smaller, until the distance to reach destination is less than $S$, which can be completed in one cycle. Hence Algorithm~\ref{algo:gatherk} gathers all the non-faulty robots in finite time.

 It can tolerate upto $n-1$ faults for $n\geqslant3$, because all the robots are moving in their unique paths. So, even if only one robot is non-faulty that will reach the destination in finite time. It is proved that for $n=2$, gathering is impossible in \textit{SSYNC}, but for $n\geqslant3$ the robots can form a multiplicity point and gather there.
\end{proof}

\section{Our Motivation: Instantaneous Computation}\label{sec:asyncic}

In $ASYNC$ model, there can be arbitrarily large delay between the \textit{look}, \textit{compute} and \textit{move} phases. Thus a robot can move based on an outdated look state. This can cause algorithms to fail. With following example we show one limitation of \textit{ASYNC} model.

Consider the following situation in the \textit{ASYNC} model as shown in Fig.~\ref{fig:asyncic} for three time instances $t_1$, $t_2$ and $t_3$ such that $t_1<t_2<t_3$.
Say an algorithm $\phi$ outputs some robots to keep the $SEC$ fixed given a configuration as the input.
Say at $t_1$, robot $r_1$ is in its look state, while robot $r_3$ is in its compute state.
Then at $t_2$ robot $r_1$ is in its compute step while robot $r_3$ started moving, since the robots fixed by $\phi$ are $r_2$, $r_4$ and $r_5$.
Now at $t_2$, $r_2$ looks and the output of $\phi$ are $r_1$, $r_4$ and $r_5$, since $r_3$ is already moved at $t_2$ changing the configuration.
We can observe that $r_1$ and $r_2$ both are moving because $r_1$ has computed based on configuration at $t_1$ and $r_2$ has computed based on configuration at $t_2$.
Then at $t_3$ both $r_1$ and $r_2$ have moved making the $SEC$ different.
Both $r_1$, $r_2$ are moving towards the center of $SEC$ in the previous configuration while another two robots $r_6$ and $r_7$ are activated and started moving towards the center of $SEC$ in the new configuration. Thus two multiplicity points are created. This situation arises because the robot $r_1$ moved based on the configuration which is outdated.

From Theorem~\ref{theorem:impossibility}, it is clear that gathering is not possible for more than one multiplicity point present. This example shows the downside of execution of algorithm based on a outdated look state. Here when $r_1$ and $r_2$ are computing, $r_3$ has moved, hence the destination they obtained is based on a outdated look data, which leads them to a different destination.
\begin{figure}[H]
\centering
\includegraphics[width=0.85\linewidth]{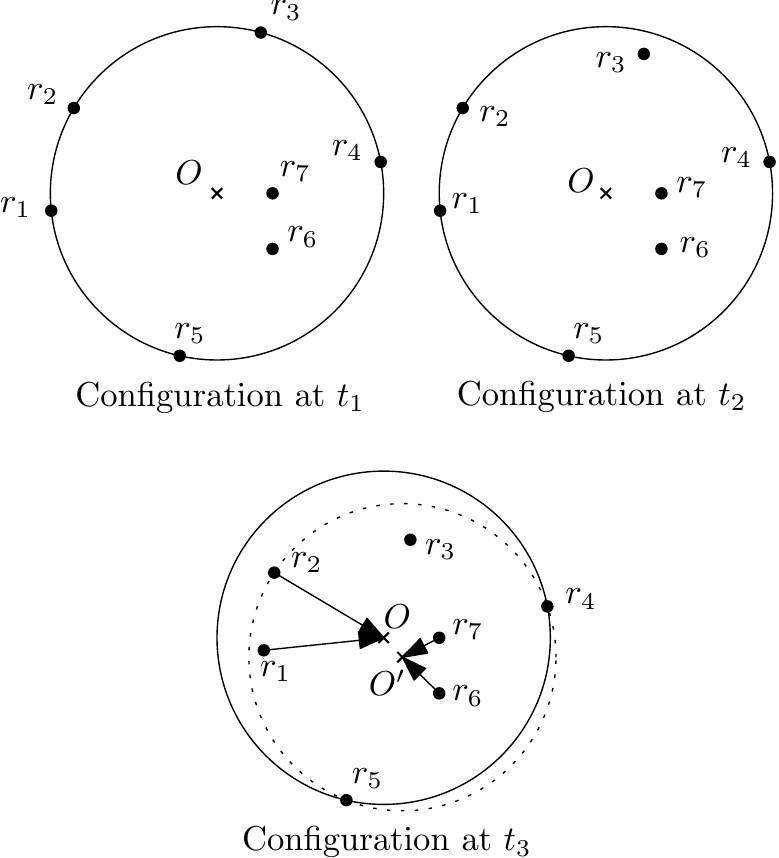}
\caption{Illustration for creation of two multiplicity in \textit{ASYNC}}\label{fig:asyncic}
\end{figure}

Now consider there is no computation delay. Then the robots would always move based on the current look information. Then situations like the previous one can be avoided. So we in this paper introduce the asynchronous model with instantaneous computation (\textit{ASYNC}$_{IC}$). In \textit{ASYNC}$_{IC
}$ the computation phase is instantaneous. This means that there is no delay between look and move.
 Any robot after completion of look stage, immediately
starts moving towards the destination. For example, if two
robots look at time $t$ and $t'$, where $t'=t+\epsilon$ for some small $\epsilon>0$,
the robot looked at time $t$ has already started its movement (unless
the destination computed is itself ) by the time the second robot looks
at $t'$. This model is denoted as \textit{ASYNC}$_{IC}$. The inactivity period
is unpredictable but finite. Since the conflict cannot happen due to compute delays, the \textit{ASYNC}$_{IC}$  model can be
considered more powerful than \textit{ASYNC}  but less powerful than
\textit{SSYNC}.

\section{Gathering ($n,\lfloor n/2\rfloor-2$) Crash Fault in ASYNC$_{IC}$ Model}
\label{sec:asyncgather}
\subsection{Model}
The robots considered in this section have exactly the same capabilities as in section~\ref{subsec:ssyncmodel}, that is they are autonomous, anonymous, homogeneous, oblivious, silent and have weak multiplicity detection. The scheduling model followed here is the \textit{ASYNC}$_{IC}$.
In this section, gathering problem is solved for $\left( n,\lfloor n/2\rfloor -2\right)$ crash fault in the \textit{ASYNC}$_{IC}$ model.
\begin{pdefinition}(($n,\lfloor n/2\rfloor -2$) Crash Fault):
Given $n$ anonymous, homogeneous, oblivious, point robots with unlimited visibility in legal initial configuration with no agreement in coordinate system, having the ability to detect multiplicity points in \textit{ASYNC}$_{IC}$ model. The objective is to achieve gathering for all non-faulty robots in ($n,\lfloor n/2\rfloor -2$) crash fault system for $n\geq 7$.
\end{pdefinition}

\subsection{Algorithm and Correctness}
The Algorithm~\ref{algo:asyncgather} (\textsc{AsyncGather($C,r_i$)}) is executed in such a fashion that the $SEC$ of the initial configuration remains as the $SEC$ of all the configurations to follow until a multiplicity point is formed.
 Initially, let $k$ robots are on the boundary and remaining $n-k$ robots are inside the $SEC$. If $k\leq n/2$, fix all the robots which are on the boundary of the $SEC$.
  So there are at least two non faulty robots which move towards the center of the $SEC$ and create a multiplicity point.
  If $k>n/2$, we divide the $SEC$ in $k$ cells as shown in Figure~\ref{fig:cell} and then use pigeonhole principle to conclude there exists at least empty cell since $n-k<k$. Among $k$ cells, let $k_1$ are empty cells. If $k_1\leq k/2$, we fix robots corresponding to the empty cells. Else if $k_1>k/2$, we fix robots corresponding to non empty cells. In any case we are fixing at most $k/2$ robots. So there are at least two non faulty robots which move towards the center of the $SEC$ and create a multiplicity point. All configurations, where fixing the $SEC$ is possible after making cells, are denoted by $C(Cell)$.
$C(Cell)$ may contain symmetric as well as asymmetric configurations.

\begin{figure}[H]\centering
\includegraphics[height=0.4\linewidth]{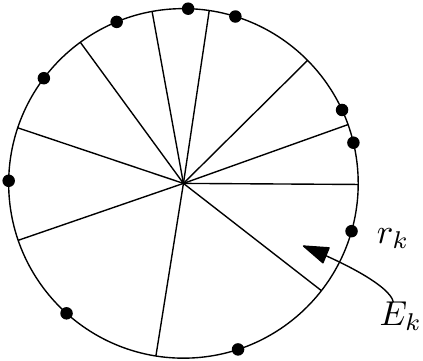}
\caption{$E_k$ is the cell of a robot $r_k$}\label{fig:cell}
\end{figure}
There are configurations for which fixing the $SEC$ is not possible using the above strategy. These configurations have the same number of robots in each cell
since robots may lie on cell boundary and consequently each cell may contain equal number of robots.
For robots on the boundary of multiple cells, those are counted as equally shared between the cells. Say a robots is shared between two cells then it contributes $0.5$ to each cell's number of robots. Say $v_i$ be the number of robots present in each cell. The possible configurations with all the cells having same $v_i$ where $k>n/2$ and no multiplicity point present are following.

\begin{figure}[H]\centering
\includegraphics[width=\linewidth]{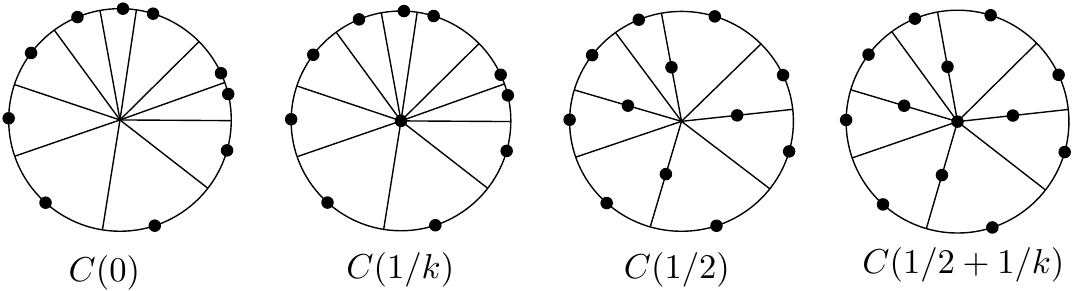}
\caption{Four configurations with the same $v_i$, where $v_i<1$}\label{fig:allsamevi}
\end{figure}

\begin{itemize}
  \item $C(0)$ : $v_i =0 \forall i \in \{1,2,\cdots,n\}$.
  \item $C(1/k)$: $v_i =1/k \forall i \in \{1,2,\cdots,n\}$ i.e., k robots are on the boundary and one robot is at the center.
  \item $C(1/2)$: $v_i =1/2 \forall i \in \{1,2,\cdots,n\}$.
  \item $C(1/2+1/k)$: Combination of $C(1/k)$ and $C(1/2)$ configuration.
\end{itemize}

Illustrations of the above configurations are shown in Figure~\ref{fig:allsamevi}.
These configurations can be symmetric as well as asymmetric.
If any of the $C(0)$, $C(1/k)$, $C(1/2)$ or $C(1/2+1/k)$ configuration is asymmetric then we can elect a leader. Consequently we can fix the $SEC$ according to Algorithm~\ref{algo:FixSEC}.
 We use the leader election algorithm presented in \cite{ChaudhuriM15}, which deterministically elects a leader for any asymmetric configuration. We show that these configurations can only be the initial configuration. If any robot moves according to our algorithm from any of these configuration then it reaches $C(Cell)$. Once the configuration is in $C(Cell)$, it never reaches to any of the four configurations stated above. Henceforth asymmetric configurations are represented by $C^*(\, )$.
 If the initial configuration is symmetric as well as any of the $C(0)$, $C(1/k)$, $C(1/2)$ or $C(1/2+1/k)$, our algorithm cannot gather since leader election is not possible. So these symmetric $C(0)$, $C(1/k)$, $C(1/2)$ or $C(1/2+1/k)$ are not included in the admissible initial configuration.

For a given configuration $C = \{p_1,p_2,$ $\cdots,p_n\}$, the $SEC$ of $C$ can be represented by a set $C^b$ of points which lie on the $SEC$. The $SEC$ of a set of points can be represented with a subset of those points. If all the other robots, which are not part of that subset move inside the $SEC$, the $SEC$ remains the same. The idea behind Algorithm~\ref{algo:FixSEC} (\textsc{FixSEC($F,C$)}) is to choose a subset of robots such that the $SEC$ remains invariant as long as the robots in that subset do not move.
\begin{observation}\label{obs:minrobotSEC}
In the set $C^b$ either there are two points which form a diameter or there are at least three points which do not lie on a semicircle.
\end{observation}
According to Observation~\ref{obs:minrobotSEC}, Algorithm~\ref{algo:FixSEC} can always choose two or three robots for fixing the $SEC$.
 The proof of Algorithm~\ref{algo:asyncgather} unfolds in the following lemmas and a theorem. Lemma~\ref{lem:fixsec} proves that we can always choose a set of robots to keep the $SEC$ fixed for any configuration.
Then Lemma~\ref{lem:onemultiplicity} shows that Algorithm~\ref{algo:asyncgather} does not create additional multiplicity points using movement strategy in section~\ref{subsec:ssyncgather}.
Lemma~\ref{lem:finitemult} shows a state diagram where the configuration with a single multiplicity point can be achieved from any configuration. Finally Theorem~\ref{theorem:asyncgather} concludes the proof of correctness.

\begin{algorithm}[!h]
\caption{\textsc{FixSEC($F$,$C$)}}\label{algo:FixSEC}
\textbf{Input:} A set of robots $F$ and configuration $C$\\
\textbf{Output:} A set of robots $C^b_{fix}$.\\
\eIf{all robots in $F$ are on the same semicircle}{
	\eIf{$|F| =1$}{
		Let $F = \{r_i\}$\\
		\eIf{$r_j$ is diametrically opposite to $r_i$}
		{$C^b_{fix} = \{r_i,r_j\}$}
		{		Find the closest robot $r_j$ and $r_k$ diametrically opposite to $r_i$ on either side of diameter passing through $r_i$. \\
		$C^b_{fix} = \{r_i,r_j,r_k\}$}
	}{
	Let $r_i,r_j$ be the farthest two robots in $F$\\
	Find the point of intersection of cord bisector of $p_ip_j$ with boundary diametrically opposite of those two robots.\\
	Find the robots $r_p$ and $r_q$ on the boundary nearest to the intersection point on either side.\\
	$C^b_{fix} = \{r_i,r_j,r_p,r_q\}$
	}
}{
	$C^b_{fix} =F$
}
\end{algorithm}

\begin{lemma}\label{lem:fixsec}
Algorithm~\ref{algo:FixSEC} determines a subset $C^b_{fix}$ of $C^b$ such that SEC of $C^b_{fix}$ is the SEC of $C$.
\end{lemma}
\begin{proof}
According to Algorithm~\ref{algo:FixSEC} for different cases we are choosing a subset of three (or four) robots of $C^b$ which do not lie on a semicircle. So to prove the lemma we need to prove that any subset of three robots of $C^b$ which do not lie on a semicircle have the same $SEC$ as $C$.

Now let us consider any three robots on the boundary of $SEC$ of $C$, say $r^b_1$, $r_2^b$ and $r^b_3$. As shown in Figure~\ref{fig:minsec}, $r^b_3$
lies outside the semicircles $r_1^br^b_2A$ and $r_2^br_1^bB$.
Consider $r^b_1$ and $r^b_2$.
\begin{figure}[H]\centering
\includegraphics[height=0.45\linewidth]{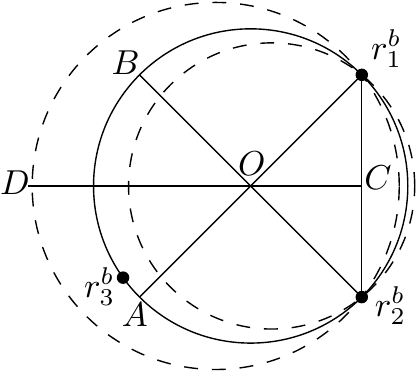}
\caption{$SEC$ of 3 robots not in a semicircle}\label{fig:minsec}
\end{figure}
All the circles that pass through these two robots have their center on the line $CD$, the perpendicular bisector of the line joining $r^b_1$ and $r^b_2$. From Figure~\ref{fig:minsec}, we can see that the smallest circle covering the three robots is the circle passing through them, which is the $SEC$ of $C$. Now the $SEC$ of the new configuration remains the same as the other robots other than the fixed robots move inside, because the $SEC$ of fixed robots is the $SEC$ of whole configuration.
\end{proof}
\begin{algorithm}
\caption{\textsc{AsyncGather($C$,$r_i$)}}\label{algo:asyncgather}
\SetKwInOut{Input}{Input}\SetKwInOut{Output}{Output}
\Input{A configuration $C$ and a robot $r_i$}
\Output{Destination of robot and its path}
\eIf{there is a multiplicity point $m$}
{$p^* = m$\\
$C^b_{fix} =\phi$\\
$distance = full$\\
\eIf{there is a robot free path towards $p^*$}
{$style = straight$}
{$style = circular$}}
{$p^* = O$, where $O$ is the center of $SEC$\\
Let $C^b=\{r_{1}^b,r^b_2,\cdots,r^b_k\}$ be the set of $k$ robots on the boundary of $SEC$.\\
\eIf{$k\leq n/2$}{
	$C^b_{fix} = C^b$\\
	$style = straight$\\
	$distance = full$
}{
Make cells \{$E_1,E_2,\cdots,E_k$\} by joining center and midpoint of arc between two consecutive robots on perimeter.\\
Find the number of robots $v_i$ in each cell $E_i$.\\
\eIf{
	all $v_i$ are same
}{
	 $r_l = $\textsc{ElectLeader($C$)}\\
	$C^b_{fix} =$\textsc{FixSEC($\{r_l\},C$)}
}
{
$V = max(v_i)$\\
$F = \{r^b_i|v_i = V\}$ \\
\eIf{$|F| \leq k/2$}{
		$C^b_{fix}=$ \textsc{FixSEC($F$,$C$)}
	}{
		$C^b_{fix} =$ \textsc{FixSEC($C^b\setminus F$,$C$)}
	}
}
\eIf{$r_i$ is on the angle bisector of its two neighbours on the boundary or there is no robot free path towards $p^*$}
	{$style= circular$}{$style = straight$}
\eIf{movement of $r_i$ to $O$ makes all $v_i$ same without creating multiplicity point}
	{$distance =  half$}{$distance=full$}
}}
\eIf{$r_i \in C^b_{fix}$}{Do not move}
	{\textsc{MoveToDest($p_i,p^*,style,distance$))}}

\end{algorithm}
\vspace{1cm}
\begin{lemma}\label{lem:onemultiplicity}
Algorithm~\ref{algo:asyncgather} gathers the non-faulty robots without creating any additional multiplicity point.
\end{lemma}
\begin{proof}
According to Algorithm~\ref{algo:asyncgather}, the robots move towards the center of $SEC$. The $SEC$ remains invariant until multiplicity point is formed. So the destination point remains invariant. Any robot $r_i$, moving towards the destination, moves according to Algorithm~\ref{algo:movetodest}, either moves in a straight or circular path.
From Lemma~\ref{lem:destgather} the straight and circular paths only intersect at the destination point. So it will not create another multiplicity point other than the destination point.
\end{proof}

\begin{lemma}\label{lem:finitemult}
Algorithm \ref{algo:asyncgather} creates a multiplicity point from an admissible initial configuration without multiplicity point in finite time.
\end{lemma}
\begin{proof}
For a configuration without a multiplicity point, Algorithm~\ref{algo:asyncgather} sets the destination as center of $SEC$ for all the robots. At any point of time, our algorithm fixes some robots such that there are at least two non-faulty robots which can move to the center of $SEC$ to create a multiplicity point.
For the configuration $C(Cell)$, we can fix robots accordingly from the cells with the highest number of robots or its complement with boundary robots, whichever is minimum.
 In Figure~\ref{fig:statediagram}, we can see the transition of configurations. From $C(Cell)$ in finite time we will reach the configuration $C(Mult)$, which contains a multiplicity point.  As the distance to the destination is the radius of $SEC$, which is finite, a multiplicity point is formed in finite time.
\begin{figure}[!h]\centering
\includegraphics[width=\linewidth]{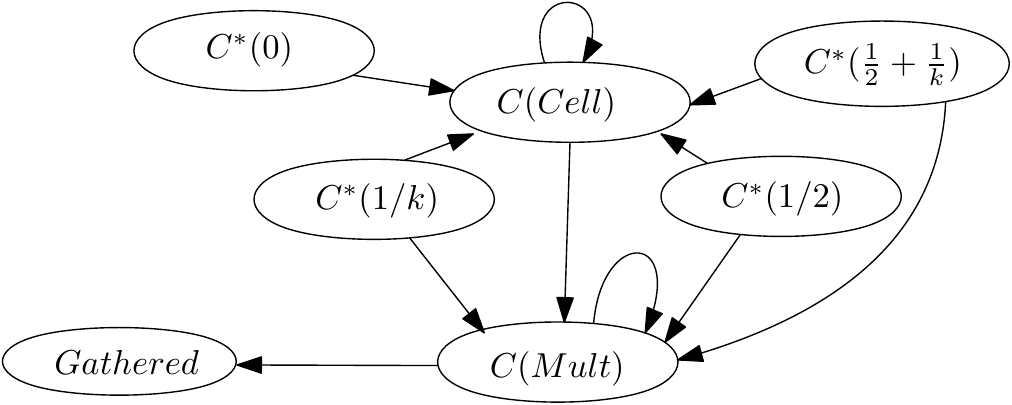}
\caption{Transitions of Configurations }
\label{fig:statediagram}
\end{figure}

$C^*(0)$, $C^*(1/k)$, $C^*(1/2)$ and $C^*(1/2+1/k)$ can only appear in the initial configuration. Then we are fixing the robots in these configuration by electing a leader on the boundary and fixing corresponding robots using Algorithm~\ref{algo:FixSEC}. From any configuration to $C^*(1/k)$ transition is not possible, because the robot would move only half of the distance towards center of $SEC$. Also from any configuration to $C^*(1/2)$ is not possible, as the robot on angle bisector of two neighbors on boundary moves away from the line. So any transition to $C^*(1/2+1/k)$ is also not possible. Similar arguments can be given for transitions not shown in Figure~\ref{fig:statediagram}.
Hence Algorithm~\ref{algo:asyncgather} creates a multiplicity point in finite time from a configuration without multiplicity point.
\end{proof}
\begin{theorem}\label{theorem:asyncgather}
Algorithm~\ref{algo:asyncgather} gathers all non-faulty robots in finite time for $n\geqslant7$.
\end{theorem}
\begin{proof}
Lemma~\ref{lem:onemultiplicity} shows that the robots always move towards the destination point, which remains invariant throughout the execution of the algorithm, without creating additional multiplicity on the way.
Lemma~\ref{lem:finitemult} presents a graph showing transition of configurations, which is a directed acyclic graph leading to the single multiplicity point configuration and from there to the gathered configuration.
The algorithm tolerates upto $\lfloor n/2\rfloor -2$ faults, because at most $n/2$ robots are fixed in the case where equal or more robots are inside the $SEC$.
Creation of multiplicity is required to make the destination point invariant, so at least two non-fixed robots must move.
So it is not possible to have all the robots inside to be faulty. Hence the fault tolerance is at $\lfloor n/2\rfloor -2$. Also Algorithm~\ref{algo:FixSEC} fixes at most 4 robots, therefore $\lceil n/2\rceil\geq4$. So $n\geq 7$.
 Thus Algorithm~\ref{algo:asyncgather} gathers all non-faulty robots in finite time.
\end{proof}
\section{Conclusion}\label{sec:conclusion}
In this paper, we have solved the problem of gathering in finite time for $(n,n-1)$ crash fault in \textit{SSYNC} model without any assumption on the coordinate system. This answers the open question in \cite{BramasT15} for fault-tolerant gathering in  \textit{SSYNC} model. In \textit{ASYNC}$_{IC}$ model under similar setting, we have solved the gathering problem with at most $\lfloor n/2\rfloor-2$ faulty robots.

We think the \textit{ASYNC}$_{IC}$ model is a step towards the positive direction. There are many problems which are still unsolved  or proved to be unsolvable in \textit{ASYNC} model can be addressed in the \textit{ASYNC}$_{IC}$ model. Many other problems like pattern formation, flocking, etc., can also be addressed in  \textit{ASYNC}$_{IC}$ model.
\bibliography{bib}

\begin{thebibliography}{10}

\bibitem{agmon2006fault}
Noa Agmon and David Peleg.
\newblock Fault-tolerant gathering algorithms for autonomous mobile robots.
\newblock {\em SIAM J. on Computing}, 36(1):56--82, 2006.

\bibitem{AugerBCTU13}
C{\'{e}}dric Auger, Zohir Bouzid, Pierre Courtieu, S{\'{e}}bastien Tixeuil, and
  Xavier Urbain.
\newblock Certified impossibility results for byzantine-tolerant mobile robots.
\newblock In {\em Stabilization, Safety, and Security of Distributed Systems -
  15th Intl. Symposium, {SSS} 2013, Osaka, Japan, November 13-16, 2013.
  Proceedings}, pages 178--190, 2013.

\bibitem{Bhagat201650}
S.~Bhagat, S.~Gan Chaudhuri, and K.~Mukhopadhyaya.
\newblock Fault-tolerant gathering of asynchronous oblivious mobile robots
  under one-axis agreement.
\newblock {\em Journal of Discrete Algorithms}, 36:50 -- 62, 2016.
\newblock \{WALCOM\} 2015.

\bibitem{Bouzid0T13}
Zohir Bouzid, Shantanu Das, and S{\'{e}}bastien Tixeuil.
\newblock Gathering of mobile robots tolerating multiple crash faults.
\newblock In {\em {IEEE} 33rd Intl. Conference on Distributed Computing
  Systems, {ICDCS} 2013, 8-11 July, 2013, Philadelphia, Pennsylvania, {USA}},
  pages 337--346, 2013.

\bibitem{BramasT15}
Quentin Bramas and S{\'{e}}bastien Tixeuil.
\newblock Wait-free gathering without chirality.
\newblock In {\em Structural Information and Communication Complexity - 22nd
  International Colloquium, {SIROCCO} 2015, Montserrat, Spain, July 14-16,
  2015, Post-Proceedings}, pages 313--327, 2015.

\bibitem{ChaudhuriM15}
Sruti~Gan Chaudhuri and Krishnendu Mukhopadhyaya.
\newblock Leader election and gathering for asynchronous fat robots without
  common chirality.
\newblock {\em J. Discrete Algorithms}, 33:171--192, 2015.

\bibitem{cieliebak2012distributed}
Mark Cieliebak, Paola Flocchini, Giuseppe Prencipe, and Nicola Santoro.
\newblock Distributed computing by mobile robots: Gathering.
\newblock {\em SIAM J. Comput.}, 41(4):829--879, 2012.

\bibitem{Cohen2005}
Reuven Cohen and David Peleg.
\newblock Convergence properties of the gravitational algorithm in asynchronous
  robot systems.
\newblock {\em SIAM J. Comput.}, 34(6):1516--1528, June 2005.

\bibitem{cohen2006convergence}
Reuven Cohen and David Peleg.
\newblock Convergence of autonomous mobile robots with inaccurate sensors and
  movements.
\newblock In {\em Proc. STACS}, pages 549--560. Springer, 2006.

\bibitem{DefagoP0MPP16}
Xavier D{\'{e}}fago, Maria~Gradinariu Potop{-}Butucaru, Julien Cl{\'{e}}ment,
  St{\'{e}}phane Messika, Philippe~Raipin Parv{\'{e}}dy, and Philippe~Raipin
  Parv{\'{e}}dy.
\newblock Fault and byzantine tolerant self-stabilizing mobile robots gathering
  - feasibility study -.
\newblock {\em CoRR}, abs/1602.05546, 2016.

\bibitem{FlocchiniPSW99}
Paola Flocchini, Giuseppe Prencipe, Nicola Santoro, and Peter Widmayer.
\newblock Hard tasks for weak robots: The role of common knowledge in pattern
  formation by autonomous mobile robots.
\newblock In {\em Proc. ISAAC}, pages 93--102, 1999.

\bibitem{HonoratPT14}
Anthony Honorat, Maria Potop{-}Butucaru, and S{\'{e}}bastien Tixeuil.
\newblock Gathering fat mobile robots with slim omnidirectional cameras.
\newblock {\em Theor. Comput. Sci.}, 557:1--27, 2014.

\bibitem{kim1995leader}
Tai~Woo Kim, Eui~Hong Kim, Joong~Kwon Kim, and Tai~Yun Kim.
\newblock A leader election algorithm in a distributed computing system.
\newblock In {\em Future Trends of Distributed Computing Systems, IEEE Intl.
  Workshop}, pages 0481--0485. IEEE Computer Society, 1995.

\bibitem{neiger1990automatically}
Gil Neiger and Sam Toueg.
\newblock Automatically increasing the fault-tolerance of distributed
  algorithms.
\newblock {\em J. of Algorithms}, 11(3):374--419, 1990.

\bibitem{Prencipe2007}
Giuseppe Prencipe.
\newblock Impossibility of gathering by a set of autonomous mobile robots.
\newblock {\em Theor. Comput. Sci.}, 384(2-3):222--231, October 2007.

\bibitem{stoller2000leader}
Scott~D Stoller.
\newblock Leader election in asynchronous distributed systems.
\newblock {\em IEEE Transactions on Computers}, 49(3):283--284, 2000.

\bibitem{suzuki1999distributed}
Ichiro Suzuki and Masafumi Yamashita.
\newblock Distributed anonymous mobile robots: Formation of geometric patterns.
\newblock {\em SIAM J. on Computing}, 28(4):1347--1363, 1999.

\end{thebibliography}
\end{document}